\documentclass[aps,pra,twocolumn,showpacs,superscriptaddress,groupedaddress,nofootinbib,longbibliography]{revtex4-1} 
\usepackage{amssymb,amsmath,amsthm}   
\usepackage{xypic}
\newcommand{\bra}[1]{\ensuremath{\langle #1 |}}
\newcommand{\ket}[1]{\ensuremath{| #1 \rangle}}
\newcommand{\inprod}[2]{\ensuremath{\langle #1 \mid #2 \rangle}}
\newtheorem{theorem}{Theorem}

\usepackage{color}
\definecolor{myurlcolor}{rgb}{0,0,0.4}
\definecolor{mycitecolor}{rgb}{0,0.5,0}
\definecolor{myrefcolor}{rgb}{0.5,0,0}
\usepackage[pdftex]{hyperref}
\hypersetup{colorlinks,
linkcolor=myrefcolor,
citecolor=mycitecolor,
urlcolor=myurlcolor}

\begin{document}

\title{Quantum theory realises all joint measurability graphs}
\author{Chris Heunen} \affiliation{Department of Computer Science, University of Oxford}
\author{Tobias Fritz} \affiliation{Perimeter Institute for Theoretical Physics}
\author{Manuel L. Reyes} \affiliation{Department of Mathematics, Bowdoin College}
\date{\today}

\begin{abstract}
  Joint measurability of sharp quantum observables is determined pairwise, and so can be captured in a graph. We prove the converse: any graph, whose vertices represent sharp observables, and whose edges represent joint measurability, is realised by quantum theory. 
  This leads us to show that it is not always possible to use Neumark dilation to turn unsharp observables into sharp ones with the same joint measurability relations,
  highlighting a caveat in the ``church of the larger Hilbert space''.	
\end{abstract}

\maketitle

\subsection*{Introduction}

One of the characteristic features of quantum theory is that not every two observables can be measured jointly.
This raises the question: what rules govern the relationship of joint measurability between quantum observables?
In this article, we prove that we can label the vertices of any given graph with sharp quantum observables in such a way that two observables are jointly measurable precisely when their vertices are connected by an edge. 
This leads us to a shortcoming of the methodology of ``the church of the larger Hilbert space'', which holds that any quantum operation can be regarded as unitary evolution of a larger, \emph{dilated}, system, and in particular that any unsharp quantum observable can be regarded as a sharp one on a dilated system.
The caveat is that dilation does not respect joint measurability.

The latter result is important to be aware of for quantum information theorists, whose bread and butter is dilation~\cite{keyl:quantuminformation};
in particular, unsharp quantum observables are used in 
quantum state discrimination~\cite{ivanovic:statediscrimination,dieks:statediscrimination,peres:statediscrimination},
photonic qubit measurement~\cite{brandt:povms},
quantum state tomography~\cite{scott:sicpovms},
quantum cryptography~\cite{applebydangfuchs:sicpovms},
and remote state preparation~\cite{killoranetal:remotestatepreparation}.
The former result is of foundational interest in its own right.
Joint measurability plays a pivotal role in \emph{contextuality}, the phenomenon that the result of measuring an observable depends on which other observables it is measured jointly with.
It has given rise to Gleason's theorem~\cite{busch:gleason}, 
Bell's inequalities~\cite{bell:inequality,CHSH:inequality,wolfperezgarciafernandez:jointmeasurability},
the Kochen--Specker theorem~\cite{kochenspecker:hiddenvariables},
Hardy's paradox~\cite{hardy:paradox},
GHZ impossibility results~\cite{GHZ:bell},
  and generalised probabilistic theories~\cite{barrett:probabilistic,popescurohrlich:boxes}.
All of these are under active study, see \textit{e.g.}~\cite{liangspekkenswiseman:specker,abramskyhardy:bell,abramskybrandenburger:sheaves}\footnote{Abramsky and Brandenburger~\cite{abramskybrandenburger:sheaves} derive abstract Kochen--Specker, Bell, Hardy, and GHZ results ``without any presupposition of quantum mechanics''. Our results could be interpreted as strengthening this approach by showing that it fully captures such ``characteristic mathematical structures of quantum mechanics, such as complex numbers, Hilbert spaces, operator algebras, or projection lattices'', after all. 
See especially Section~7.1, which discusses~\cite{lovaszsaksschrijver:orthogonal}. That paper has similar results as this article, but with orthogonality instead of joint measurability, requiring extra conditions.}.
In particular, there are (non)contextuality inequalities that are violated by quantum mechanics and hence can be used to experimentally detect quantum effects~\cite{klyachkocanbinicioglushumovsky:inequality}, that come from graph theory~\cite{cabelloseveriniwinter:graphs,fritzleverriersainz:contextuality}.

\subsection*{Realisation as yes-no questions}

Let $G$ be a graph. Write $v,w,x,y, \ldots \in G$ for its vertices, and $v \sim w$ when $v$ and $w$ are connected by an edge. By convention, we agree that $v \sim v$ for any vertex $v$.
Think of the vertices as observables, that are jointly measurable precisely when they are connected by an edge.

We will be concerned with several kinds of observables: all will be particular types of structures on a Hilbert space, but what joint measurability means will vary. By a \emph{realisation} of $G$ as observables on a Hilbert space $H$, we mean a function $x \mapsto O_x$ that sends vertices to observables in such a way that $O_x$ and $O_y$ are jointly measurable if and only if $x$ and $y$ are connected by an edge.
As the basic step, we will first consider \emph{yes-no questions}, that is, \emph{projections}. A set of projections is defined to be jointly measurable when each pair in it commutes. We now prove that any graph is realisable as projections on some Hilbert space.


\begin{theorem}\label{thm:yesno}
  Any graph has a realisation as projections on some Hilbert space.
\end{theorem}
\begin{proof}
First, consider the special case of a graph $G_{v,w}$ where all pairs of vertices are connected by an edge, except for two fixed vertices $v,w$ that are not connected.
Fix two projections on $\mathbb{C}^2$ that do not commute, for example:
\[
  \ket{0}\bra{0} = \begin{pmatrix} 1 & 0 \\ 0 & 0 \end{pmatrix},
  \quad
  \ket{+}\bra{+} = \tfrac{1}{2} \begin{pmatrix} 1 & 1 \\ 1 & 1 \end{pmatrix}.
\]
We can use these to build a realisation of $G_{v,w}$ as projections on $\mathbb{C}^2$.
Define $p_v = \ket{0}\bra{0}$, $p_w = \ket{+}\bra{+}$, and $p_x=0$ for all other vertices $x \neq v,w$.
By construction, all pairs $p_x$ and $p_y$ for vertices $x,y \in G_{v,w}$ commute, except for $p_{v}$ and $p_{w}$. Hence $x \mapsto p_x$ realises $G_{v,w}$ as projections on $\mathbb{C}^2$. We will denote the dependency on $v$ and $w$ of this realisation by writing $p_x^{v \not\sim w}$ for $p_x$.

Now that we know how to obstruct a single pair of vertices from being jointly measurable, we return to an arbitrary graph $G$. 
Let the Hilbert space $H = \bigoplus_{v \not\sim w} \mathbb{C}^2$ be the direct sum of copies of $\mathbb{C}^2$, where the direct sum ranges over all pairs of vertices that are not connected by an edge.
For any vertex $x \in G$, then $p_x = \bigoplus_{v \not\sim w} p_x^{v \not\sim w}$ gives a well-defined projection on $H$~\cite{heunen:ltwo}.
Now, if $x \sim y$, then all $p_x^{v \not\sim w}$ and $p_y^{v \not\sim w}$ commute by construction, and so $p_x$ and $p_y$ commute.
Similarly, if $x \not\sim y$, then $p_x^{x \not\sim y}$ and $p_y^{x \not\sim y}$ do not commute, and so $p_x$ and $p_y$ do not commute.
All in all, we have constructed a realisation $x \mapsto p_x$ of $G$ as projections on $H$.
\end{proof}

If $f \colon G_1 \to G_2$ is an injective function between graphs satisfying $f(v) \not\sim_2 f(w)$ when $v \not\sim_1 w$, then the realisations are related by $p_{x} = V^\dag p_{f(x)} V$ for an isometry $V$.

\subsection*{Dimension bounds for yes-no questions}

There is a well-defined \emph{minimal} dimension in which a graph with $V$ vertices can be realised as projections.
The construction in the proof of Theorem~\ref{thm:yesno} showed that this minimal dimension is at most $2N$, where $N$ is the number of non-edges, \textit{i.e.}\  pairs of vertices that are not connected by an edge.
Notice that Theorem~\ref{thm:yesno} makes sense for graphs of arbitrary size; if the graph is infinite, then the number $N$ should be regarded as a cardinal number. In particular, the theorem implies that finite graphs can be realised as projections on a finite-dimensional Hilbert space, namely in dimension $2N$.
Clearly $N\leq\tfrac{|G|(|G|-1)}{2}$, so that the minimal dimension is at most $|G|(|G|-1)$; this inequality is saturated for graphs without edges, for which $N=\tfrac{|G|(|G|-1)}{2}$.

We will now show that the minimal dimension that any graph can be realised in is at most $|G|$.

\begin{theorem}\label{thm:bounds}
  Any graph has a realisation as projections on a Hilbert space whose dimension is at most $|G|$, the number of vertices of $G$.
\end{theorem}
\begin{proof}
  If $|G|$ is an infinite cardinal number, then it equals $|G|(|G|-1)$, and the claim follows from the above considerations.

  We may therefore assume that the graph is finite. 
  Consider the Hilbert space of dimension $|G|+N$, with orthonormal basis vectors $\ket{x}$ for each vertex $x \in G$ and $\ket{\{v,w\}}$ for each non-edge $v \not\sim w$. For each vertex $x \in G$, define a vector $\ket{\psi_x} =  \ket{x} + \sum_{x \not\sim v} \ket{\{x,v\}}$,
  where the sum ranges over all vertices $v$ not adjacent to $x$.
  For distinct vertices $x$ and $y$ then
  \[
    \inprod{\psi_x}{\psi_y} = \begin{cases} 0, & x \sim y, \\ \inprod{\{x,y\}}{\{y,x\}} = 1, & x \not\sim y. \end{cases}
  \]
  Thus $\ket{\psi_x}$ and $\ket{\psi_y}$ are orthogonal when $x \sim y$, but not orthogonal or parallel when $x \not\sim y$ (because $\inprod{\psi_x}{\psi_y}^2 = 1 < 2 \cdot 2 \leq \inprod{\psi_x}{\psi_x} \inprod{\psi_y}{\psi_y}$).

  Letting $p_x$ be the projection onto $\ket{\psi_x}$ constructs a realisation $x \mapsto p_x$ as projections.
  Finally, notice that each $p_x$ has rank 1. So we may restrict the Hilbert space down to just the linear span of the $|G|$ vectors $\ket{\psi_x}$.
  This restricts the realisation $x \mapsto p_x$ to a Hilbert space of dimension at most $|G|$.
\end{proof}

The construction in the proof relied on the fact that projections onto single vectors commute precisely when the vectors are parallel or orthogonal. This is closely related to \emph{orthogonal representations} of graphs, which have been studied in the literature~\cite[Sec.~9.3]{lovasz:geometricrepresentations}. For example, if the complement of the graph is connected after removing any $V-d-1$ vertices, then one can assign unit vectors in $\mathbb{R}^d$ to the vertices such that all these vectors are different, and two vectors are orthogonal if and only they share an edge. In general, if we insist that the projections $p_x$ have rank one, then the minimal dimension in which the complement of the ``path'' graph 
\[\xymatrix@1{\bullet \ar@{-}[r] & \bullet \ar@{-}[r] & \cdots \ar@{-}[r] & \bullet}\]
can be realised is $|G|-1$~\cite{hogben:minimumrank,boothetal:minimumrank}. In that sense, Theorem~\ref{thm:bounds} is very close to being optimal. We leave open the question of whether allowing $p_x$ to have higher rank can lead to more efficient realisations.

\subsection*{Realisation as sharp observables}

The above results easily extend from yes-no questions to \emph{sharp observables}, that is, \emph{projection valued measures} (PVMs). A PVM is a set $P$ of mutually orthogonal projections that sum to 1. A family $P_1,P_2,\ldots$ of PVMs is jointly measurable when $p$ and $q$ commute for all $p \in P_i$ and $q \in P_j$ and all $i,j$~\cite{buschgrabowskilahti:operational}. Hence a specification of sharp quantum observables and which ones are jointly measurable is determined pairwise, and can also be captured in a graph. 

\begin{theorem}\label{thm:pvm}
  Any graph has a realisation as PVMs on a Hilbert space whose dimension is at most the number of vertices.
\end{theorem}

\begin{proof}
  Given a graph $G$ with vertices $x,y,\ldots$, simply replace the projection $p_x$ of Theorem~\ref{thm:bounds} by the PVM $P_x = \{p_x,1-p_x\}$: the PVMs $P_x$ and $P_y$ are jointly measurable if and only if $p_x$ and $p_y$ commute.
\end{proof}


In the joint measurability graph of all projections on a Hilbert space, a special role is played by \emph{maximal cliques}: maximal sets of vertices, every two of which are connected by an edge. They correspond to PVMs $P$ that are maximally fine-grained, in the sense that all $p \in P$ have rank one. More precisely, given such a PVM $P$, the set of all projections commuting with all $p \in P$ form a maximal clique.
Conversely, a maximally fine-grained PVM can be recovered as the minimal projections in a maximal clique.\footnote{Given a maximal clique of projections in a Hilbert space, the C*-algebra it generates is commutative. Therefore it has a commutative projection lattice. By maximality, this lattice coincides with the clique, which is therefore a Boolean sublattice of the full projection lattice.}

It is not always possible to realise a graph as projections in a way that sends maximal cliques to PVMs. For a counterexample, consider the ``fork'' graph with three vertices and two edges,
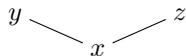
\[\xymatrix@C-1ex@R-5ex{
  y && z \\
  & x \ar@{-}[ur] \ar@{-}[ul] 
}\]
Suppose there were a realisation as projections with $p_x+p_y=1=p_x+p_z$. Then $p_y = 1-p_x = p_z$, making $p_x$ and $p_y$ commute, contradicting the fact that $y \not\sim z$.
We leave open the interesting question of characterising which graphs can be realised as projections in a way that sends maximal cliques to PVMs.


We call a realisation as projections $x \mapsto p_x$ \emph{faithful} when distinct vertices $x \neq y$ give rise to distinct projections $p_x \neq p_y$.
The previous example might have given pause to the reader who intuitively expected a realisation as projections of a graph to be
faithful. 
The construction of Theorem~\ref{thm:yesno} might not be faithful, because
vertices $x \in G$ that are connected to all others end up being realised by the projection $p_x=0$ commuting with anything. Any realisation as projections can be made faithful as follows. 
Enlarge the Hilbert space to $H \oplus H'$, where $H'$ has orthonormal basis $\{ |x \rangle \mid x \in G\}$, and send $x \in G$ to $p_x \oplus \ket{x}\bra{x}$.
This is clearly faithful, and has the same commutativity properties as the original realisation.


We can similarly extend to realisations as sharp quantum observables that are \emph{not dichotomic}.
If the vertices $x \in G$ are labeled with numbers $n_x \geq 2$, we can realise the graph as PVMs such that $P_x$ has $n_x$ elements. Enlarge the Hilbert space to $H \oplus \bigoplus_{x \in G} H_x$, where $H_x$ has orthonormal basis $\{ |3_x\rangle, \ldots, |n_x\rangle\}$, and send $x \in G$ to 
$P_x = \{ p_x \oplus 0, (1-p_x) \oplus 0 \} \cup \{ 0 \oplus |i \rangle \langle i | \mid i=3,\ldots,n_x\}$. This has the same commutativity properties as the original realisation.


In principle, one could imagine physical theories in which joint measurability of observables is not determined pairwise. (Indeed, we will see shortly that unsharp observables in quantum mechanics form a case in point.) To model joint measurability, we then have to generalise to \emph{hypergraphs}, in which a hyperedge can connect any number of vertices~\cite{cabelloseveriniwinter:graphs,fritzleverriersainz:contextuality}. Any graph induces a hypergraph, where a set of vertices forms a hyperedge when every two vertices in it are connected by an edge. Our definition of \emph{realisability} easily carries over to hypergraphs: vertices still represent observables, and a set of vertices forms a hyperedge precisely when it is jointly measurable. Combining the above results with the well-known fact that sharp observables are jointly measurable when they commute~\cite{kraus:states,ludwig:foundations},
we obtain the following characterisation: a hypergraph is realisable as sharp quantum observables if and only if it is induced by a graph.

Just as we have discussed dimension bounds for the realisations of graphs by projections, we can also ask what the minimal dimension is to realise a given graph by PVMs. As the proof of Theorem~\ref{thm:pvm} shows, any realisation as projections can be turned into a realisation as PVMs, and hence the PVM minimal dimension is at most the projection minimal dimension. As witnessed by the multitude of proofs of the Kochen--Specker theorem in $\mathbb{C}^3$ and $\mathbb{C}^4$~\cite{pavicicetal:graphapproach}, there is quite a lot of ``room'' already in these low dimensions, and and one may wonder whether this is already enough to realise every graph as a PVM. This turns out not to be the case.

\begin{theorem}
There is no dimension $d$ in which all graphs can be realised as PVMs.
\end{theorem}

\begin{proof}
For a given $d$, we construct a graph which cannot be realised in dimension $d$ as follows. Let $B_d$ be the number of partitions of $\{1,\ldots,d\}$; this is the $d$th Bell number. Now take a graph with $B_d+1$ vertices designated as ``action'' vertices and $n:=\lceil \log_2(B_d+1) \rceil$ many ``control'' vertices. Enumerate the action vertices using bitstrings of length $n$. Then, action vertex $v$ shares an edge with control vertex number $k$ if and only if the $k$-th bit in the bitstring associated to $v$ is $1$. Also, every two action vertices share an edge, while two control vertices may or may not share an edge.

This graph cannot be realised in dimension $d$: since every action vertex is connected to a different set of control vertices, no two action vertices can map to the same PVM. On the other hand, all these PVMs must be jointly measurable, and hence all their elements can be diagonalised in the same basis. In this fixed basis, every PVM therefore corresponds to a partition of $\{1,\ldots,d\}$. But since we have $B_d+1$ many PVMs, which is higher than the number of partitions of $\{1,\ldots,d\}$, this is impossible.
\end{proof}

\subsection*{Unsharp observables and Neumark dilation}

We now turn to the most general kind of \emph{(unsharp) quantum observables}, namely \emph{positive operator valued measures} (POVMs). These are defined as functions $E$ from some \emph{outcome space} $I$ to operators on a Hilbert space that are bounded between $0$ and $1$ and form a resolution of the identity\footnote{While we only consider \emph{discrete} POVMs here, all our results hold unabated for positive-operator valued measures on the Borel sets on a compact Hausdorff space, by reformulating them in the language of C*-algebras and completely positive maps; see \textit{e.g.}~Theorem~4.6 of~\cite{paulsen:completelypositive}.}: $\sum_{i \in I} E(i) = 1$, and $0 \leq E(i) \leq 1$ for each $i \in I$.
If $E(i)$ is a projection for each $i$, we actually have a PVM. Therefore we may also write $P(i)$ instead of $p_i$ for PVMs $P=\{p_i \mid i \in I\}$. 
A family of POVMs $E_1,E_2,\ldots$ is defined to be jointly measurable when there exists a joint POVM $E$ of which they are the marginals: if POVMs $E_n$ have outcome space $I_n$, then $E$ should have outcome space $\prod_n I_n$ and satisfy
\begin{align*}
  E_1(i_1) = \sum_{i_2 \in I_2, i_3 \in I_3, \ldots} E(i_1,i_2,i_3,\ldots), \\
  E_2(i_2) = \sum_{i_1 \in I_1, i_3 \in I_3, \ldots} E(i_1,i_2,i_3,\ldots),
\end{align*}
and so on~\cite{buschgrabowskilahti:operational, heinosaarireitznerstano:jointmeasurability}. This reduces to the previously considered notions of joint measurability for yes-no questions and sharp quantum observables.


Neumark's famous \emph{dilation} theorem says that any POVM can be dilated to a PVM on a larger Hilbert space,
or in other words, that any POVM is the \emph{compression} of a PVM on a larger Hilbert space: 
if $E$ is a POVM on a Hilbert space $H$ with outcome space $I$, then there exist a Hilbert space $K$, an isometry $V \colon H \to K$, and a PVM $P$ on $K$ with outcome space $I$, such that $E(i) = V^\dagger P(i) V$~\cite{buschgrabowskilahti:operational,paulsen:completelypositive}.
This forms an important part of the philosophy that John Smolin called ``the church of the larger Hilbert space'', which holds that one need not care about unsharp observables as long as ancilla spaces are taken into account.


There is an extension of Neumark's dilation theorem for families of observables. We call a family $E_1,E_2,\ldots$ of POVMs, with outcome spaces $I_1,I_2,\ldots$, on a Hilbert space $H$ \emph{jointly dilatable} when there exist a Hilbert space $K$, an isometry $V \colon H \to K$, and a single PVM $P$ with outcome space $\prod_n I_n$ such that $E_n(i) = \sum_j V^\dagger P(i,j) V$, where $j$ ranges over $\prod_{m \neq n} I_m$, and we write $(i,j)$ for the obvious element\footnote{That is, for $i \in I_n$ and $j \in \prod_{m \neq n} I_m$, the element $(i,j) \in \prod_n I_n$ has $n$th component $i$ and other
components given by the components of $j$.} of $\prod_m I_m$. It is now a matter of unfolding definitions to prove that a family of POVMs on a Hilbert space is jointly measurable if and only if it is jointly dilatable.


We now study how joint measurability behaves under Neumark dilation.
Suppose POVMs $E_1$, $E_2$, and $E_3$ are compressions of PVMs $P_1$, $P_2$, and $P_3$
with respect to different isometries.
If $\{E_1,E_2,E_3\}$ are jointly measurable there is a \emph{single} isometry $V_{123}$ that dilates the joint POVM $E_{123}$ (to, say, $P_{123}$).
But if the $E_i$ are merely pairwise jointly measurable, then there exist PVMs $P_{ij}$ and three isometries $V_{ij}$ that dilate $E_{ij}$ to $P_{ij}$.
What we will show is that even if one has all three pairwise dilations $P_{ij}$ via isometries $V_{ij}$ at hand, it may be the case that there is no triplewise dilation $P_{123}$ via any isometry $V_{123}$. 
Thus Neumark dilation cannot always turn unsharp observables into sharp ones with the same joint measurability relations.
In this sense, Neumark dilation does not reflect joint measurability.

\begin{theorem}\label{thm:povms}
  There is a family $\{E_n\}$ of POVMs on a Hilbert space $H$ 
  that does not allow 
  an isometry $V \colon H \to K$ and
  a family of PVMs $\{P_n\}$ (with the same outcome spaces as $E_n$) on $K$ 
  with $E_n(i) = V^\dag P_n(i) V$
  in such a way that a subset of $\{E_n\}$ is jointly measurable if and only if 
  the corresponding subset of $\{P_n\}$ is jointly measurable. 
\end{theorem}

\begin{proof}
Perhaps the simplest counterexample starts with a family $\{E_1,E_2,E_3\}$ of POVMs on the Hilbert space $H=\mathbb{C}^2$, every pair of which is jointly measurable, but which is not jointly measurable itself~\cite{kraus:states,ludwig:foundations,heinosaarireitznerstano:jointmeasurability,wolfperezgarciafernandez:jointmeasurability,liangspekkenswiseman:specker}. Its hypergraph is a ``hollow triangle''.
\[\xymatrix@C-1ex@R-3ex{
  E_2 \ar@{-}[rr] && E_3\\
  & E_1 \ar@{-}[ur] \ar@{-}[ul] 
}\]
In other words, this (hyper)graph is realisable as POVMs.

In contrast, as noted above,
joint measurability of PVMs is determined pairwise, which will lead us to a contradiction. 
Suppose PVMs $\{P_1,P_2,P_3\}$ as in the statement of the theorem existed.
Then, by our assumptions,
the pairwise joint measurability of the $E_n$ would imply pairwise joint measurability of the $P_n$, so the $P_n$ would necessarily be triplewise jointly measurable as well.
In other words, then the $\{E_n\}$ would be (triplewise) jointly dilatable.
But this contradicts the fact that the $\{E_n\}$ are not (triplewise) jointly measurable.
In summary: joint measurability of the
putative $P_n$ would imply
joint measurability of the $E_n$, since a joint POVM can be constructed as the compression of a joint PVM.
\end{proof}

We could interpret the previous theorem as a warning against an unreflected belief in ``the church of the larger Hilbert space''. If you care about (non-)joint measurability of observables, you cannot simply ignore unsharp quantum observables in favour of their dilated sharp observables, even if ancilla spaces are taken into account, and you have to take the unsharpness involved seriously. 

This plays a role in quantum protocols that rely on unsharp observables that are not jointly measurable, in which case the usual analysis by dilation to sharp observables should not be used.
For example, \cite{brandt:povms} explicitly constructs a PVM implementation of a POVM and mentions that this ``faithfully represents the POVM''. However, 
PVM implementations cannot always represent joint measurability relations within families of POVMs. We suspect that it may be possible to turn this apparent problem into a \emph{feature} which can be exploited in new quantum information protocols. More concretely, we imagine situations in which a number of parties share some quantum information resource, but only certain subgroups of these parties are allowed joint access to it.

\begin{acknowledgments}
Thanks to Rob Spekkens for discussions on joint dilatability and the proof of Theorem~\ref{thm:povms}, and to David Roberson for pointing out an issue with an earlier proof of Theorem~\ref{thm:bounds}.
The first author was supported by the Engineering and Physical Sciences Research Council.
Research at Perimeter Institute is supported by the Government of Canada through Industry Canada and by the Province of Ontario through the Ministry of Economic Development and Innovation. The second author has been supported by the John Templeton Foundation.
\end{acknowledgments}

\bibliographystyle{IEEEtran}  
\bibliography{pairwise}

\end{document}